\documentclass[english,11pt]{elsarticle}
\usepackage{graphicx}
\usepackage{epstopdf}
\usepackage{ae,aecompl}

\usepackage[T1]{fontenc}
\usepackage[latin9]{inputenc}
\usepackage[letterpaper]{geometry}
\usepackage{setspace}
\usepackage{amsmath}
\usepackage{amsthm}
\usepackage{amssymb}
\usepackage{bm}       
\usepackage{bbm}
\usepackage{dcolumn}  
\usepackage{psfrag}
\usepackage{pstricks}
\usepackage{pst-node}

\makeatletter

\newcommand{\maximize}{\mathop{\operator@font{maximize}}}
\newcommand{\minimize}{\mathop{\operator@font{minimize}}}

\newtheorem{thm}{Theorem}
\newcounter{definition}

\makeatother

\begin{document}
 
\title{Mass conserved elementary kinetics is sufficient for the existence
of a non-equilibrium steady state concentration}

\author{R. M. T. Fleming
     \\ Center for Systems Biology,
        University of Iceland
     \\ Ph: +354 618 6245, Email: ronan.mt.fleming@gmail.com
     \\ Sturlugata 8, Reykjavik 101, Iceland.
     \\[8pt] I. Thiele
     \\ Faculty of Industrial Engineering, Mechanical Engineering \& Computer Science,
     \\ Sturlugata 8, Reykjavik 101, Iceland.
     \\ University of Iceland}

\date{\today}

\begin{abstract}
Living systems are forced away from thermodynamic equilibrium by exchange of mass and energy with their environment. In order to model a biochemical reaction network in a non-equilibrium state one requires a mathematical formulation to mimic this forcing. We provide a general formulation to force an arbitrary large kinetic model in a manner that is still consistent with the existence of a non-equilibrium steady state. We can guarantee the existence of a non-equilibrium steady state assuming only two conditions; that every reaction is mass balanced and that continuous kinetic reaction rate laws never lead to a negative molecule concentration. These conditions can be verified in polynomial time and are flexible enough to permit one to force a system away from equilibrium. In an expository biochemical example we show how a reversible, mass balanced perpetual reaction, with thermodynamically infeasible kinetic parameters, can be used to perpetually force a kinetic model of anaerobic glycolysis in a manner consistent with the existence
of a steady state. Easily testable existence conditions are foundational for efforts to reliably compute non-equilibrium steady states in genome-scale biochemical kinetic models.
  \smallskip

\end{abstract}

\maketitle

\section{Introduction}

There are various approaches for simulation of biochemical network
function \citep{klipp2009systems}. In principle, an ideal approach
would accurately represent known physicochemical principles of reaction
kinetics, tailored with kinetic parameters specific to a particular
organism. However, when modeling genome-scale biochemical networks,
one's choice of modeling approach is also shaped by concerns of computational
tractability. One of the main reasons that flux balance analysis \citep{Watson1986,Fell1986,Savinell1992c,Pal06,orth2010flux}
has found widespread applications in genome-scale modeling is that
the underlying algorithm is typically based on linear optimization.
In general, industrial quality software implementations of linear
optimization algorithms are guaranteed to find an optimal solution,
if one exists, or otherwise give a certificate that the problem posed
is infeasible. In the process of iterative model development, one
may use flux balance analysis to test if a stoichiometric model, obtained
from a draft reconstruction, actually admits a steady state flux (reaction
rate) \citep{thieleTests}. If not, various algorithms have been developed
that help to detect missing reactions \citep{thieleTests} or \emph{stoichiometric
inconsistencies} \citep{gevorgyan2008detection} (discussed further
in Section \ref{sub:Stoichiometry}).

Flux balance analysis predicts fluxes that satisfy steady state mass
conservation but not necessarily energy conservation or the second
law of thermodynamics \citep{beard2002eba,looplaw,fleming2008stk}.
Whilst the set of steady state mass conserved fluxes includes those
that are thermodynamically feasible, additional constraints are required
in order to guarantee a flux that additionally satisfies energy conservation
and the second law of thermodynamics \citep{fleming2009opi}. Within
the set of steady state fluxes that satisfy mass conservation, energy
conservation and the second law of thermodynamics, there is a subset
that additionally satisfy various reaction kinetic rate laws \citep{cook2007eka}.
The satisfaction of such kinetic constraints is important for accurate
representation of various biochemical processes where the abundance
of a molecule affects the rate of a reaction, e.g., allosteric regulation
\citep{heinrich1996rcs,Jamshidi2008a}.

There are various algorithmic barriers to genome-scale kinetic modeling
that preclude satisfaction of all of the aforementioned thermodynamic
and kinetic constraints, without resorting to rate law approximation.
Apart from the open algorithmic challenge to develop an algorithm
for efficient computation of thermodynamically and kinetically feasible
steady states in genome-scale kinetic models, there is also the challenge
of mathematically expressing the necessary and sufficient conditions
for existence of such a steady state in a manner that can be efficiently
tested given a genome-scale model.

The quantitative study of chemical reaction kinetics has a long history,
beginning perhaps with Ludwig Wilhemy's 1850 discovery that the rate
of a chemical reaction is proportional to the concentrations of consumed
substrates \citep{wilhelmy1850Kinetics}. This fundamental law of
chemical kinetics is well known to chemists and biochemists alike.
However, due to imperfect inheritance of knowledge, there are other
useful facts, established generations ago, that are slipping from
the consciousness of chemists and biochemists alike - as measured by
contemporary citations. Even for a multifaceted paper with many citations,
those citations can be due to a historically facet, not necessarily
the most useful facet in a contemporary sense. A case in point is
the 1962 paper by James Wei entitled ``Axiomatic treatment of chemical
reaction systems''  \citep{wei1962axiomatic}. This excellent paper
has been cited 65 times, but only twice in the last 10 years and infrequently
by theoretical biochemists. The vast majority of citations refer to
Wei's treatment of the stability of chemical reaction systems with
Lyapunov functions. Exceptionally, the importance of Wei's result,
concerning the conditions for existence of non-equlibrium steady states,
is realized, e.g., in 1976 Perelson used Wei's result to illustrate
the danger inherent in concluding results from mathematical models
of systems of chemical reactions that do not conserve mass \citep{Perelson1976}.

Herein we build upon Wei's work that establishes sufficient conditions
for the existence of at least one steady state reaction flux and molecule
concentration for a broad class of chemical kinetic models. This class
of kinetic model includes all those networks with exclusively mass
balanced chemical reactions, where the kinetics rate laws are such
that molecule concentration can never be a negative quantity. This
class of reaction network includes all biochemical reaction networks.
The conditions for existence may be easily tested by a trivial check
on the kinetic rate law formulation for each reaction, together with
a test of stoichiometric consistency using linear optimization \citep{gevorgyan2008detection}. 

In Section \ref{sec:Chemical-reaction-networks} we introduce some
mathematical definitions of pertinent chemical reaction network concepts.
Section \ref{sec:ExistenceSS} states and proves theorem concerning
the existence of a non-negative steady state molecule concentration
(vector) for mass conserved elementary reaction kinetics. This theorem
and proof follow the more general case outlined in broad strokes by
Wei \citep{wei1962axiomatic}. Based on citation history, Wei's result
seems to be lost upon the biochemical modeling community. In Section
\ref{sec:UtilitySSexistence}, we illustrate for the first time, the
utility of Wei's existence theorem for modeling non-equlibrium steady
states with various examples, including a kinetic model of anaerobic
glycolysis in \emph{Trypanosoma brucei}. Finally we summarise and
attempt to place this work in the context of established mathematical
approaches to model biochemical reaction networks.

\section{Chemical reaction networks}

\label{sec:Chemical-reaction-networks}

\subsection{\label{sub:Stoichiometry}Stoichiometry}

Consider a biochemical network with $m$ molecules and $n$ \emph{elementary}
reactions. An elementary reaction is one for which no reaction intermediates
have been detected or need to be postulated in order to describe the
chemical reaction on a molecular scale. It follows that the reaction
stoichiometry is sufficient to define the molecularity of the molecules
involved in the reaction. One may combine elementary reactions together
to form a \emph{composite} reaction. One can define the topology of
the resulting \emph{hypergraph} using a generalized incidence matrix,
$\mathbf{S}\in\mathbb{\mathbb{Z}}^{m,n}$, where $\mathbf{S}$ is
always singular and typically $r\equiv\textrm{rank}(\mathbf{S})<m<n$
for large biochemical networks. Each row in this \emph{stoichiometric
matrix} represents a particular molecule, e.g., glucose, whilst each
column represents a reversible biochemical reaction. We assume that
all biochemical reactions are indeed \emph{reversible} \citep{lewis1925new}.
For each reversible reaction, convention dictates one direction be
designated \emph{forward} and the other \emph{reverse}. With respect
to the forward direction, for all $i=1\ldots m$ and $j=1\ldots n$,
$S_{ij}<0$ if molecule $i$ is a \emph{substrate} in a reaction,
meaning that it is consumed by the reaction $j$, $S_{ij}>0$ if molecule
$i$ is a \emph{product}, meaning that it is produced by a reaction,
and $S_{ij}=0$ otherwise. Typically stoichiometric coefficients are
integers reflecting the whole number molecularity for a molecule consumed
or produced in a reaction.

Each column of a stoichiometric matrix contains at least one negative
coefficient and one positive coefficient, reflecting either the chemical
conversion of one molecule to another, or in multi-compartmental models,
the transport of a molecule from one compartment to another, i.e.,
a transport reaction may consume one molecule in a reactant compartment
and produces one molecule in a different product compartment, even
if the molecule is physically identical. We assume that each column
of $\mathbf{S}$ corresponds to one \emph{mass conserving} chemical
reaction. A necessary, but insufficient condition for mass balancing
is that each column of $\mathbf{S}$ must have at least one positive
coefficient and at least one negative coefficient. We say that a chemical
reaction is $linear$ when the corresponding column of $\mathbf{S}$
contains two nonzero coefficients, $\{-1,1\}$. We say that a chemical
reaction is $bilinear$ when the corresponding column of $\mathbf{S}$
contains three non-negative coefficients, $\{-1,1,1\}$ or $\{-1,-1,1\}$.
There may be more than one negative (positive) coefficient in a column
when a reaction involves more than one substrate (product). In reaction
networks with \emph{composite} reactions, nonzero stoichiometric coefficients
are typically not of magnitude one. However, even the most complicated
composite reaction can be decomposed into linear and bilinear reactions.

Each row of $\mathbf{S}$ contains at least one positive coefficient
and at least one negative coefficient, reflecting the requirement
for at least one reaction to produce and at least one reaction to
consume each molecule. A stoichiometric matrix for a chemical reaction
network is said to be \emph{consistent} if each molecule can be assigned
a single positive molecular mass, without violating mass conservation,
and \emph{inconsistent} otherwise \citep{gevorgyan2008detection,Famili2003}.
Mathematically, this translates to the existence of at least one strictly
positive vector, $\mathbf{m}\in\mathbb{R}_{>0}^{m}$, in the left
nullspace of a consistent stoichiometric matrix $\mathbf{S}^{T}\cdot\mathbf{m}=\mathbf{0}$.
Strictly, we could say that each row of $\mathbf{S}$ corresponds
to an isotopically distinct molecular entity in order that the corresponding
molecular mass be precisely defined, as two otherwise identical molecules
can have different molecular mass depending on their isotopic label
e.g. $^{13}\textrm{C}$ vs $^{14}\textrm{C}$ glucose. However, we
shall assume that reaction kinetic parameters are isotopomer invariant
so we need not be so strict.

\subsection{Reaction kinetics}

Perhaps the simplest reaction kinetic assumption is that a unidirectional
reaction rate is proportional to the product of the concentrations
of each substrate consumed \citep{wilhelmy1850Kinetics}. Let us define
forward and reverse stoichiometric matrices, $\mathbf{F},\mathbf{R}\in\mathbb{R}_{\ge0}^{m,n}$
respectively, where $\mathbf{F}_{ij}$ denotes the stoichiometry of
substrate $i$ in forward reaction $j$ and $\mathbf{R}_{ij}$ denotes
the stoichiometry of substrate $i$ in reverse reaction $j$. It follows
that the stoichiometric matrix is defined by $\mathbf{S}\equiv-\mathbf{F}+\mathbf{R}$.
It is possible for the same molecule to appear as both a substrate
and a product in the same unidirectional reaction, e.g., an auto-catalytic
reaction, so it is natural to define $\mathbf{S}$ in terms of $\mathbf{F}$
and $\mathbf{R}$, rather than the other way around. We we may now
express \emph{elementary kinetics} for forward and reverse reaction
rates, respectively $\mathbf{v}_{f},\mathbf{v}_{r}\in\mathbb{R}^{n}$,
as 
\begin{equation}
\begin{array}{cc}
\mathbf{v}_{f}(\mathbf{k}_{f},\mathbf{x})\equiv\textrm{diag}(\mathbf{k}_{f})\cdot\exp(\mathbf{F}^{T}\cdot\ln(\mathbf{x})),\\
\mathbf{v}_{r}(\mathbf{k}_{r},\mathbf{x})\equiv\textrm{diag}(\mathbf{k}_{r})\cdot\exp(\mathbf{R}^{T}\cdot\ln(\mathbf{x})),
\end{array}\label{eq:massActionKinetics}
\end{equation}
 where we assume non-negative \emph{elementary kinetic parameters}
$\mathbf{k}_{f},\mathbf{k}_{r}\in\mathbb{R}_{\geq0}^{n}$, and the
exponential or natural logarithm of a vector is meant component-wise
\footnote{Strictly, it is not proper to take the logarithm of a unit that has
physical dimensions. This difficulty can be avoided by considering
$\mathbf{x}$ as a vector of mole fractions rather than concentrations
(Eq. 19.93 in \citep{physChemBRR}). %
}. We say that a pair of forward and reverse elementary kinetic parameters
are \emph{thermodynamically feasible} when they satisfy

\begin{equation}
\exp\left(-\mathbf{S}_{j}^{T}\cdot\frac{\mathbf{u}^{\circ}}{RT}\right)=\frac{k_{f,j}}{k_{r,j}},\label{eq:thermoFeasibleParam}
\end{equation}
 where $\mathbf{u}^{o}\in\mathbb{R}^{m}$ is a vector of standard
chemical potentials, $R$ is the gas constant and $T$ is temperature.
Thermodynamically feasible kinetic parameters for all reactions implies
detailed balance at thermodynamic equilibrium, i.e., $\mathbf{v}_{f}=\mathbf{v}_{r}$
\citep{tolman1979psm}.

We refer to \eqref{eq:massActionKinetics} as \emph{mass action kinetics}
\citep{prigogine1954chemical} only after we have stated our assumption
that \eqref{eq:thermoFeasibleParam} also holds for each mass conserving
reversible reaction. In the words of Horn \& Jackson \citep{horn1972gma}
``a kinetic description of chemical reactions in closed systems with
ideal mixtures, completely consistent with the requirements of stoichiometry
and thermodynamics, may be obtained by satisfying the following four
conditions'': 
\begin{itemize}
\item [{(a)}] The rate function of each elementary reaction is of the
mass action form. 
\item [{(b)}] The stoichiometric coefficients are such that mass is
conserved in each elementary reaction. 
\item [{(c)}] The kinetic constants in the rate functions are constrained
in such a way that the principle of detailed balancing is satisfied. 
\item [{(d)}] The stoichiometric coefficients are non-negative integers. 
\end{itemize}
Condition (a) is represented by \eqref{eq:massActionKinetics} and
conditon (b) is satisfied by strict adherence to elemental balancing
for each reaction during network reconstruction \citep{thieleTests,Thorleifsson2011}.
Horn \& Jackson \citep{horn1972gma} considered \emph{general mass
action kinetics} when only (a) was assumed to hold. We shall consider
\emph{mass conserved elementary kinetics} where (a) and (b) are assumed
to hold but (c) and (d) are allowed to be relaxed for a subset of
reactions. We shall return to this point in the discussion. We take
the dynamical equation for mass conserved elementary kinetics to be
\begin{eqnarray}
\dot{\mathbf{x}}\equiv\frac{d\mathbf{x}}{dt}=\mathbf{S}\cdot(\mathbf{K}_{f}\cdot\exp(\mathbf{F}^{T}\cdot\ln\left(\mathbf{x}\right))-\mathbf{K}_{r}\cdot\exp(\mathbf{R}^{T}\cdot\ln\left(\mathbf{x}\right)))\label{eq:massConservedElementaryKinetics}
\end{eqnarray}
 where $t$ denotes time, all reactions conserve mass, all reactions
are reversible, $\mathbf{K}_{f}=\textrm{diag}(\mathbf{k}_{f})$, $\mathbf{K}_{r}=\textrm{diag}(\mathbf{k}_{r})$
and $\mathbf{k}_{f},\mathbf{k}_{r}\in\mathbb{R}_{\geq0}^{m}$.

\subsection{Concentration non-negativity}
\begin{figure}[b]
\includegraphics[width=0.5\textwidth]{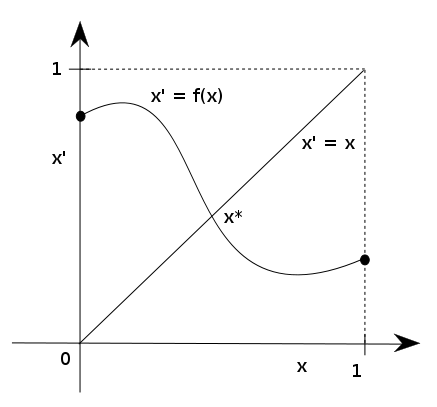}

\caption{\label{fig:Conceptual-illustration-of}Conceptual illustration of
the Brouwer fixed point theorem in one dimension. An arbitrary continuous
function $f$ is represented by the graph mapping the abscissa to
the ordinate. On the abscissa, the bounded interval between zero and
one represents the domain. The domain is closed as it includes its
endpoints and is convex as every line interval is a convex set. The
interval between zero and one on the ordinate represents the codomain
of $f$. The image of $f$ is a continuous interval contained within
the codomain of $f$, so $f$ is \emph{into.} It is unnecessary for
$f$ to be \emph{onto}. The diagonal line represents equality of the
values in the domain and codomain. Tracing a continuous curve from
left to right, between the dots, we see that it must intersect the
diagonal at some point. At the point where the diagonal meets the
curve, the value of the domain equals the value of the image $x^{*}=x'=f(x)$.
When the value passed into the function is the same as the value passed
out by the function, this value is termed a \emph{fixed point}. (Figure
adapted from http://commons.wikimedia.org/wiki/File:Fixedpoint1d.svg)}
\end{figure}

Physically, one would expect that molecular concentration be a non-negative
quantity. Starting from an initial non-negative concentration, $\mathbf{x}_{0}\geq0$,
it has been proven mathematically that all subsequent concentrations
are non-negative when the evolution of a system is subject to elementary
reaction kinetics \citep{bernstein1999nonnegativity,chellaboina2009modeling}.
Elementary kinetics is \emph{essentially non-negative} if, for all
$i=1\ldots m$, $\dot{x_{i}}\geq0$ for all $x_{i}\in\mathbb{R}_{\geq0}$,
where $\dot{x_{i}}$ and $x_{i}$ denote the $i^{th}$ component of
$\dot{\mathbf{x}}$ and $\mathbf{x}$, respectively. This mathematical
formulation can be easily be chemically interpreted. Suppose the concentration
of molecule A is zero; then irrespective of what the non-negative
concentration is for all other molecules, the rate of change in concentration
for molecule A is non-negative. To understand why, recall that in
elementary kinetics the rate of a unidirectional reaction is always
the product of a non-negative elementary kinetic parameter and the
concentration(s) of the substrate(s), each to the power of the absolute
value of the corresponding stoichiometric coefficient. If a molecule's
concentration is zero then all reactions consuming that molecule have
a rate of zero; hence, there can be no consumption of that molecule
and its rate of change in concentration is non-negative.

Consider the chemical reaction network with the single reaction
\[
A\overset{k_{f}}{\underset{k_{r}}{\rightleftharpoons}}B+C.
\]
 Let the forward and reverse reaction rates be given by elementary
kinetics $v_{f}=k_{f}a$ and $v_{r}=k_{r}bc$ with $k_{f},k_{r}\in\mathbb{R}_{\geq0}$,
where lowercase refers to concentration of the corresponding uppercase
molecule. If the initial concentrations of all molecules are non-negative
it is impossible for any molecule's concentration to go negative because
$a=0$ implies $v_{f}=0$ so no consumption of $A$ would occur and
therefore $a$ cannot be negative. Observe that the conditions for
concentration non-negativity place no constraints on the ratio of
forward over reverse kinetic parameter, but only that all kinetic
parameters be non-negative.

A further technical restriction is that $\dot{\mathbf{x}}$ be given
by a locally Lipschitz continuous function, but this is easily satisfied
by deterministic formulations of kinetics where the corresponding
differential equations are continuously differentiable (Lemma 1, \citep{chellaboina2009modeling}).
Starting from an initial non-negative concentration, equation \eqref{eq:massConservedElementaryKinetics}
constrains all subsequent concentrations to be non-negative as we
assume elementary reaction kinetics with non-negative kinetic parameters.
Even for non-negative but thermodynamically infeasible kinetic parameters,
which violate \eqref{eq:thermoFeasibleParam}, it is true that if
one begins with non-negative initial concentrations then all subsequent
concentrations remain non-negative.

\section{Existence of a steady state for mass conserved elementary kinetics}

\label{sec:ExistenceSS} The following theorem establishes sufficient
conditions for existence of a steady state concentration for a dynamical
system governed by mass conserved elementary kinetics.
\begin{thm}
\label{thm:existenceSS}Let the dynamical equation for mass conserved
elementary kinetics be \emph{
\begin{equation}
\dot{\mathbf{x}}\equiv\frac{d\mathbf{x}}{dt}=\mathbf{S}\cdot(\mathbf{K}_{f}\cdot\exp(\mathbf{F}^{T}\cdot\ln\left(\mathbf{x}\right))-\mathbf{K}_{r}\cdot\exp(\mathbf{R}^{T}\cdot\ln\left(\mathbf{x}\right)))\label{eq:elementaryKinetics}
\end{equation}
} where $\mathbf{x}\equiv\mathbf{x}(t)\in\mathbb{R}^{m}$ is molecule
concentration at time $t>0$, $\dot{\mathbf{x}}\in\mathbb{R}^{m}$
is the time derivative of concentration, $\mathbf{K}_{f}=\textrm{diag}(\mathbf{k}_{f})$,
$\mathbf{K}_{r}=\textrm{diag}(\mathbf{k}_{r})$ and $\mathbf{k}_{f},\mathbf{k}_{r}\in\mathbb{R}_{\geq0}^{n}$
are non-negative forward and reverse kinetic parameters. $\mathbf{F},\mathbf{R}\in\mathbb{R}_{\geq0}^{m,n}$
are forward and reverse stoichiometric matrices. $\mathbf{S}\equiv-\mathbf{F}+\mathbf{R}$
is a consistent stoichiometric matrix defined by the existence of
at least one strictly positive vector $\mathbf{m}\in\mathbb{R}_{>0}^{m}$,
such that $\mathbf{S}^{T}\cdot\mathbf{m}=\mathbf{0}$. Assuming a
strictly positive initial concentration $\mathbf{x}_{0}\equiv\mathbf{x}(0)\in\mathbb{R}_{>0}^{m}$,
then there exists at least one non-negative steady state concentration
$\mathbf{x}_{\ge0}^{\star}$, such that $\dot{\mathbf{x}}=\mathbf{0}$.\end{thm}
\begin{proof}
We define the function $\mathbf{f}(\mathbf{x})\,:\,\mathbb{R}^{m}\rightarrow\mathbb{R}^{m}$
\begin{eqnarray}
\mathbf{f}(\mathbf{x}) & = & \mathbf{x}+\dot{\mathbf{x}}\label{eq:functionForfixedPoint}
\end{eqnarray}
 where $\mathbf{f}(\mathbf{x})\equiv\mathbf{x}(t+\tau)$ represents
the concentration after an arbitrary small time interval $\tau>0$.
If it exists, a \emph{fixed point} $\mathbf{x}^{\star}$, such that
$\mathbf{f}(\mathbf{x}^{\star})=\mathbf{x}^{\star}$, corresponds
to a steady state concentration, $\mathbf{x}^{\star}(t)=\mathbf{x}^{\star}(t+\tau),$
or equivalently $\dot{\mathbf{x}}=0$. Observe that $f(\mathbf{x})$
is continuous as $\dot{\mathbf{x}}$ is given by a continuous function.
Let us define the closed, bounded and convex set
\[
\Omega=\left\{ \mathbf{x}\geq0,\;\mathbf{m}^{T}\cdot\mathbf{x}=\mathbf{m}^{T}\cdot\mathbf{x}_{0}>0\right\} ,
\]
 as the domain of $\mathbf{f}(\mathbf{x})$. For a strictly positive
initial concentration vector, then elementary kinetics, continuity
of $\mathbf{f}(\mathbf{x})$ and non-negative kinetic parameters are
sufficient conditions to ensure that all subsequent concentrations
are non-negative, $\mathbf{f}(\mathbf{x})\geq0$ (Theorem 2, \citep{chellaboina2009modeling}).
Since $\mathbf{S}\in\mathbb{R}^{m,n}$ is a stoichiometrically consistent
matrix there exists an $\mathbf{m}\in\mathbb{R}_{>0}^{m}$ such that
$\mathbf{m}^{T}\cdot\mathbf{S}=0$ and therefore $\mathbf{m}^{T}\cdot\dot{\mathbf{x}}=0$.
It follows that $\mathbf{m}^{T}\cdot\mathbf{f}(\mathbf{x})=\mathbf{m}^{T}\cdot\mathbf{x}_{0}$
for all $\mathbf{x}\in\Omega$. This together with the non-negativity
of $\mathbf{f}(\mathbf{x})$ establishes that $\mathbf{f}(\mathbf{x})\in\Omega$.
We have now established that $\mathbf{f}(\mathbf{x})$ is a continuous
mapping from a closed, bounded and convex set into itself. By Brouwer's
fixed point theorem there exists at least one fixed point $\mathbf{f}(\mathbf{x}^{\star})=\mathbf{x}^{\star}$
and therefore there exists at least one steady state. 
\end{proof}
In the proof of existence of steady states for mass conserved elementary
kinetics, we make use of the following theorem that we state without
proof. 
\begin{thm}
(Brouwer fixed point theorem). Let $\Omega$ be a closed, bounded
and convex set in $\mathbb{R}^{m}$, and let $\Phi:\mathbb{R}^{m}\rightarrow\mathbb{R}^{m}$
be a function that is continuous on $\Omega$ and maps $\Omega$ into
itself. Then there exists a point $\mathbf{x}\in\Omega$ such that
$\Phi(\mathbf{x})=\mathbf{x}$. 
\end{thm}
There is a voluminous literature on fixed point theory \citep{granas2003fixed}.
Figure \ref{fig:Conceptual-illustration-of} illustrates an intuitive
appreciation for the rationale behind Brouwer's fixed point theorem
by considering a one dimensional case, that is, a function mapping
of an interval on a line \emph{into} itself.

\section{Utility of steady state existence theorem}

\label{sec:UtilitySSexistence}

Theorem \ref{thm:existenceSS} is non-constructive in the sense that
it does not describe an algorithm for computation of steady state
concentrations. In the case where one is modeling a system of exclusively
mass conserved reversible reactions with mass action kinetics, it
has long been known that a unique steady state concentration can be
computed with a single convex optimization problem \citep{white1958chemical}.
Such a steady state corresponds to a thermodynamic equilibrium where
detailed balance holds. The development of a reliable algorithm to
compute non-equilibrium steady states for arbitrary large networks
is an important open problem. In the process of algorithm development,
it is essential to know, \emph{a priori}, if at least one steady state
exists. Otherwise it becomes impossible to distinguish if a failure
to compute a steady state is due to a shortcoming of an algorithm's
design, or due to to an ill-posed problem without a solution in the
first instance.

The key difference between an equilibrium and non-equilibrium steady
state is that the latter is accompanied by thermodynamic forcing of
the system by the environment. In chemical reaction networks, time
invariant thermodynamic forcing has been mathematically represented
by clamping a subset of concentrations away from equilibrium or injecting
mass across the boundary of the model \citep{QB05}. However, for
kinetic modeling of time invariant concentrations, any formulation
of a forced system must be compatible with the existence of at least
one steady state concentration vector. It is therefore important to
establish, if possible, the conditions for existence of at least one
steady state for each formulation. Some formulations are actually
incompatible with the existence of a steady state.

\subsection{System forcing where a steady state may not exist}

One approach to forcing a system is to represent the exchange of molecules
between a system and its environment with a set of mass imbalanced
source or sink reactions, respectively $\emptyset\rightarrow A$ and
$A\rightarrow\emptyset$, where $A$ is an arbitrary molecule. Let
$\mathbf{S}_{e}\in\mathbb{Z}^{m,k}$ denote the stoichiometry of mass
imbalanced exchange reactions. To the author's knowledge, for networks
with bilinear reactions, no theorem exists that defines the conditions
on the data $\{\mathbf{S},\mathbf{S}_{e},\mathbf{k}_{f},\mathbf{k}_{r}\}$
such that there still exists at least one non-equilibrium steady state.
As described in Section \ref{sub:Stoichiometry}, an augmented stoichiometric
matrix, $\bar{\mathbf{S}}=[\begin{array}{cc}
\mathbf{S} & \mathbf{S}_{e}\end{array}]$, containing mass imbalanced exchange reactions will not be stoichiometrically
consistent, so Theorem \ref{thm:existenceSS} does not apply.

Another approach to forcing a system, in which all reactions are mass
balanced, is to attempt to iterate toward a steady state of the forced
dynamical system
\begin{equation}
\dot{\mathbf{x}}\equiv\frac{d\mathbf{x}}{dt}=\mathbf{S}\cdot(\mathbf{v}_{f}(\mathbf{k}_{f},\mathbf{x})-\mathbf{v}_{r}(\mathbf{k}_{r},\mathbf{x}))-\mathbf{b},\label{eq:bForcing}
\end{equation}
 where $\mathbf{b}$ is a concentration invariant forcing vector in
the range of the stoichiometric matrix, $\mathbf{b}\in\mathcal{R}(\mathbf{S})$.
If one chooses a $\mathbf{b}^{\star}\in\mathcal{R}(\mathbf{S})$ such
that 
\begin{equation}
\mathbf{S}\cdot(\mathbf{v}_{f}(\mathbf{k}_{f},\mathbf{x}^{\star})-\mathbf{v}_{r}(\mathbf{k}_{r},\mathbf{x}^{\star}))=\mathbf{b}^{\star}\label{eq:steadyState_bStar}
\end{equation}
is satisfiable, then this would correspond to forcing in a manner
independent of molecule concentration. Given $\mathbf{x}^{\star}$
it is trivial to compute $\mathbf{b}^{\star}$ but not the other way
around. If we assume that unidirectional reaction rates are as defined
in \eqref{eq:massActionKinetics} and \eqref{eq:thermoFeasibleParam},
then by rearrangement, one may express \eqref{eq:massConservedElementaryKinetics}
as
\[
\dot{\mathbf{x}}=[\begin{array}{cc}
\mathbf{S} & -\mathbf{S}]\end{array}\cdot\textrm{diag}(\left[\begin{array}{c}
\mathbf{k}_{f}\\
\mathbf{k}_{r}
\end{array}\right])\cdot\exp([\begin{array}{cc}
\mathbf{F} & \mathbf{R}\end{array}]^{T}\cdot\ln\left(\mathbf{x}\right)),
\]
where $[\begin{array}{cc}
\mathbf{F} & \mathbf{R}\end{array}]\in\mathbb{Z}^{m,2n}$. Typically $m<n$ and $\textrm{rank}([\begin{array}{cc}
\mathbf{F} & \mathbf{R}\end{array}])<n$, so the image of $[\begin{array}{cc}
\mathbf{F} & \mathbf{R}\end{array}]^{T}\cdot\ln\left(\mathbf{x}\right)$ is not the whole of $\mathbb{R}^{2n}$ and therefore the set of all
$\dot{\mathbf{x}}$ is a subset of the range of the stoichiometric
matrix. The set of $\mathbf{b}^{\star}$ such that \eqref{eq:steadyState_bStar}
is satisfiable is only a subset of the range of the stoichiometric
matrix, so a steady state may not necessarily exist for an arbitrary
$\mathbf{b}$. Attempting to force a system with \eqref{eq:bForcing}
leaves one with the problem of attempting an \emph{a priori} choice
of concentration invariant forcing vector that may not admit a steady
sate concentration.

\subsection{System forcing where a steady state must exist}

Theorem \ref{thm:existenceSS} is constructive in the sense that it
leads to a method to force a system in a manner that ensures there
always exists at least one non-equilibrium steady state concentration
vector. The key point is to recognise that Theorem \ref{thm:existenceSS}
holds for any choice of non-negative kinetic parameters. Additional
thermodynamic constraints on kinetic parameters \eqref{eq:thermoFeasibleParam}
are optional on a per reaction basis. If all reversible reactions
have thermodynamically feasible kinetic parameters, then the only
steady state is thermodynamic equilibrium, but if at least one reversible
reaction is modeled with thermodynamically infeasible kinetic parameters,
that violate \eqref{eq:thermoFeasibleParam}, then detailed balance
does not hold but there always exists at least one non-equilibrium
steady state.

It is not physicochemically realistic to model actual chemical reactions
with thermodynamically infeasible kinetic parameters for any form
of kinetic rate law \citep{cook2007eka}. However, one may include
a mass balanced, reversible \emph{perpetireaction} with thermodynamically
infeasible kinetic parameters, purely for the modeling purpose of
forcing a system away from equilibrium. (The prefix \emph{perpeti}
is from the latin \emph{perpes} meaning lasting throughout, continuous,
uninterrupted, continual, perpetual; see www.perseus.tufts.edu.).
The augmentation of a consistent stoichiometric matrix with a mass
balanced perpetireaction still retains the stoichiometric consistency
of the augmented matrix. Assuming that elementary reaction kinetics
is used to model each unidirectional reaction there will still exist
a steady state. We now illustrate one choice of perpetireaction by
considering a biochemical example.

\subsubsection{\label{sub:Tbrucei}Mass conserved elementary kinetics of \emph{Trypanosoma
brucei}}

The utility of Theorem \ref{thm:existenceSS} can be illustrated by
considering a typical kinetic modeling scenario, such as the modeling
of anaerobic glycolysis in the African trypanosome, \emph{Trypanosoma
brucei}, the causative agent of human African trypanosomiasis \citep{Barrett2010,Bakker2010}.
Based on a phenomenological kinetic model of \emph{T. brucei} glycolysis
\citep{Bakker1997} the stoichiometry of anaerobic glycolysis may
be represented in skeleton form by the composite chemical reactions
in Figure \ref{fig:The-anaerobic-glycolysis}. Modeling composite
reactions with elementary kinetic rate laws is \emph{pseudoelementary
kinetics}, but for the purpose of illustrating the utility of Theorem
\ref{thm:existenceSS}, this distinction is superfluous. Starting
with extracellular glucose the overall stoichiometry of this pathway
may be given by the mass balanced composite reaction 
\begin{equation}
\textrm{glucose }\rightleftharpoons\textrm{glycerol + pyruvate }+H^{+}.\label{eq:TrypGlycAner}
\end{equation}
This composite reaction may be used as a perpetireaction that, in
reverse, connects the outputs of anaerobic glycolysis back to the
glucose input. This perpetireaction is the TrypGlycAner reaction in
Figure \ref{fig:The-anaerobic-glycolysis}. As perpetireaction kinetic
parameters violate \eqref{eq:thermoFeasibleParam}, no equilibrium
steady state exists as detailed balance \citep{tolman1979psm} is
violated, but there exists at least one non-equilibrium steady state
for the augmented system. Such a non-equilibrium steady state conserves
mass in all reactions, and all except the perpetireaction conserve
energy.
\begin{figure}
\includegraphics[width=0.45\textwidth]{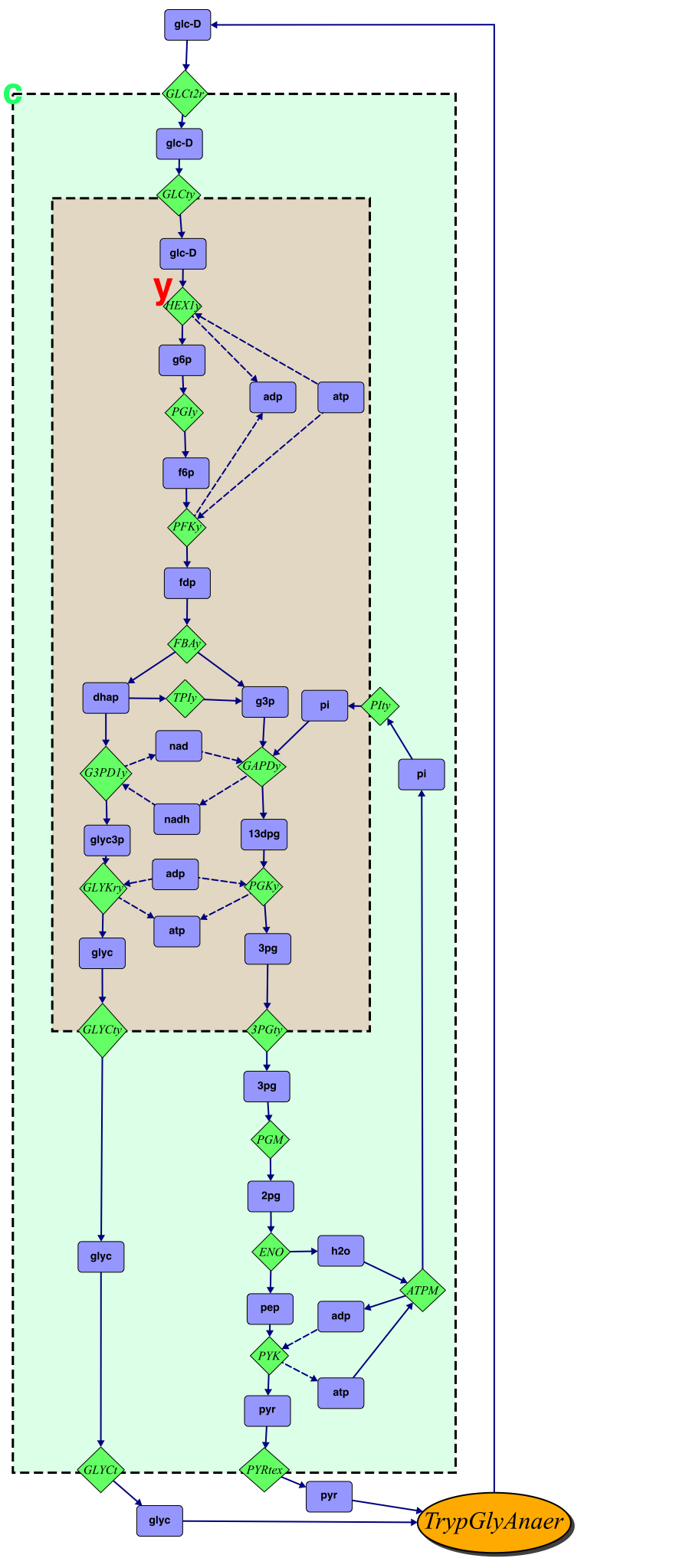}

\caption{\label{fig:The-anaerobic-glycolysis}The anaerobic glycolysis pathway
of \emph{Trypanosoma brucei}, part of which is in the cytoplasm (c)
and part of which is within a membrane-bounded, peroxisome-like organelle
termed a glycosome (y). The input to this pathway is glucose (glc-D)
and the outputs are glycerol (glyc), pyruvate (pyr) and hydrogen ion
(not shown for clarity). For modeling purposes, one can construct
a perpetual reaction, or \emph{perpetireaction} (TrypGlycAner), from
extracellular output metabolites to extracellular glucose input to
form, with the anaerobic glycolysis pathway, a stoichiometrically
balanced cycle. With appropriate choice of kinetic parameters, a non-equilibrium
steady state concentration corresponds to net flux in the directions
indicated by the arrows. Dashed arrows indicate the involvement of
cofactors. (Illustration created with \emph{Omix} \citep{Droste2011}.)}
\end{figure}

Assuming constant temperature and pressure, and uniform spatial concentrations
within a single compartment, the existence of a single chemical potential
for each (compartment specific) molecule is a necessary and sufficient
condition for conservation of energy \citep{planck1945tt,minty1960mn,ross2008taf}.
The violation of \eqref{eq:thermoFeasibleParam} by the pair of forward
and reverse elementary kinetic parameters, $p_{f}$ and $p_{r}$,
for the perpetireaction means that there exists no single standard
chemical potential for each molecule and hence no single chemical
potential for each molecule. With reference to the \emph{T. brucei}
example, one or more of glucose, glycerol, pyruvate or hydrogen ion
can not be assigned a unique chemical potential. This is equivalent
to the statement that the stoichiometrically weighted sum of chemical
potential around the single stoichiometrically balanced cycle formed
by the anaerobic glycolysis pathway and the perpetireaction is not
zero \citep{beard2002eba,fleming2008stk}. 

One can also think of the perpetireaction as a chemical reaction that
extracts energy, but not molecular moieties, from an infinitely large
source in the environment. In a non-equilibrium steady state, the
amount of energy extracted from the environment per unit time by the
perpetireaction is equal to the entropy production rate of all the
other reactions. In analogy with electrical networks, at a non-equilibrium
steady state, a perpetireaction acts like a direct current voltage
source. Indeed, in the representation of electrical networks, even
the most elementary circuit diagram forms a closed cycle with some
voltage source in the loop to drive electrons around the circuit.
In numerically calculated steady states the \emph{T. brucei} model,
if all thermodynamically feasible kinetic parameters are given a unit
value and $p_{f}>p_{r}$ then the net steady state flux of anaerobic
glycolysis proceeds in the usual direction (Figure \ref{fig:The-anaerobic-glycolysis}).
However, if $p_{f}<p_{r}$ then anaerobic glycolysis proceeds in the
reverse direction. In an electrical network analogy, this switch is
equivalent to reversing the polarity of a direct current voltage source.

\subsubsection{More general examples}

In Section \ref{sub:Tbrucei}, we considered an example where the
overall system being modeled consisted of a single stoichiometrically
balanced pathway. More precisely a single extreme ray of an augmented
stoichiometric matrix
\[
\overline{\mathbf{S}}=[\begin{array}{cc}
\mathbf{S} & \mathbf{S}_{p}\end{array}]
\]
where $\mathbf{S}_{p}$ is a column vector with reaction stoichiometry
for a single composite perpetireaction, with net forward direction
right to left in \ref{eq:TrypGlycAner}. In the general case of an
arbitrary mass balanced stoichiometric matrix $\mathbf{S}\in\mathbb{R}^{m,n}$,
then $\mathbf{S}_{p}\in\mathbb{R}^{m,k}$ can be an arbitrary set
of $k$ column vectors, each of which could be a linear basis vector
for the range of $\mathbf{S}$, accompanied by $k$ perpetireactions.
In fact, any $\mathbf{b}\in\mathcal{R}(\mathbf{S})$ can be used to
create an augmented stoichiometric matrix 
\[
\overline{\mathbf{S}}=[\begin{array}{cc}
\mathbf{S} & -\mathbf{b}\end{array}]
\]
that is still also consistent. Note however that this is not equivalent
to forcing a system like \ref{eq:bForcing} as the steady state $\mathbf{x}^{\star}$,
that we know must exist, satisfies
\[
0=[\begin{array}{cc}
\mathbf{S} & -\mathbf{S}]\end{array}\cdot\textrm{diag}(\left[\begin{array}{c}
\mathbf{k}_{f}\\
\mathbf{k}_{r}
\end{array}\right])\cdot\exp([\begin{array}{cc}
\mathbf{F} & \mathbf{R}\end{array}]^{T}\cdot\ln\left(\mathbf{x}^{\star}\right))-[\begin{array}{cc}
\mathbf{b} & -\mathbf{b}\end{array}]\cdot\textrm{diag}(\left[\begin{array}{c}
p_{f}\\
p_{r}
\end{array}\right])\cdot\exp([\begin{array}{cc}
\mathbf{b}_{f} & \mathbf{b}_{r}\end{array}]^{T}\cdot\ln\left(\mathbf{x}^{\star}\right)),
\]
where $p_{f},p_{r}>0$ are perpetireaction parameters, with $\mathbf{b}_{f}\equiv\max(-\mathbf{b},0)$
and $\mathbf{b}_{r}\equiv\max(\mathbf{b},0)$ defined in an analogous
manner to $\mathbf{F}$ and $\mathbf{R}$. Even more general is the
consideration of continuous kinetic rate laws that guarantee concentration
non-negativity. The same strategy to augment the system with perpetireactions,
yet retain stoichiometric consistency, will yield a forced system
where there always exists a steady state concentration vector.

\section{Discussion}

In the present work, Theorem \ref{thm:existenceSS} gives sufficient
conditions for the existence of at least one non-negative steady state
concentration vector, assuming elementary reaction kinetics for a
set of mass balanced chemical reactions. All kinetic parameters are
required to be positive but do not have to satisfy thermodynamic constraints
\eqref{eq:thermoFeasibleParam} on the ratio of forward over reverse
elementary kinetic parameter. Actual biochemical reactions are modeled
with reactions that have thermodynamically feasible kinetic parameters.
In order to conserve mass, yet admit a non-equilibrium steady state,
one may augment the set of thermodynamically feasible reactions with
one or more \textit{perpetireactions} (perpetual reactions), defined
as reactions with thermodynamically infeasible kinetic parameters.
This gives the flexibility to model the non-equilibrium dynamics of
a system closed to exchange of mass with the environment yet not isolated
with respect to the exchange of energy with the environment. In a
non-equilibrium steady state, the net input of chemical energy is
the driving force for net flux through the stoichiometrically balanced
system. 

Theorem \ref{thm:existenceSS} does not preclude that a subset of
molecule concentrations are actually zero at a steady state. There
may exist more than one steady state and no conclusion can be drawn
as to the stability or otherwise of the steady states that exist.
Theorem \ref{thm:existenceSS} is non-constructive in that it does
not provide an algorithm to compute a non-equilibrium steady state.
However, Theorem \ref{thm:existenceSS} makes use of Brouwer's fixed
point theorem so, assuming the conditions required for Theorem \ref{thm:existenceSS}
to hold, it may be possible to apply related constructive fixed point
theorems \citep{granas2003fixed} to design an algorithm that is guaranteed
to converge to a non-equilibrium steady state. Contributions from
fixed point theorists are encouraged and this is part of the reason
for the detail in Section \ref{sub:Stoichiometry}.

The sufficient conditions for existence of a non-equilibrium steady
state are easily tested numerically for arbitrary large chemical networks
\citep{thiele2009gcr}. As described in an elegant paper by \citet{gevorgyan2008detection}
the stoichiometric consistency of a metabolic network can be proved
or disproved by attempting the linear optimization problem
\begin{align}
\underset{\textrm{\textbf{m}}}{\textrm{minimize}}\qquad & \mathbf{e}^{T}\cdot\mathbf{m}\label{eq:consistencyTest1}\\
\textrm{such that}\qquad & \mathbf{S}^{T}\cdot\mathbf{m}=\mathbf{0}\label{eq:consistencyTest2}\\
 & \mathbf{m}>\mathbf{0}\label{eq:consistencyTest3}
\end{align}
where $\mathbf{e}$ denotes a vector of ones and $\mathbf{S}\in\mathbb{Z}^{m,n}$
is a stoichiometric matrix. If there exists an $\mathbf{m}$ satisfying
\eqref{eq:consistencyTest2} and \eqref{eq:consistencyTest3}, then
$\mathbf{S}$ is stoichiometrically consistent, otherwise a suitable
solver will provide a certificate of infeasibility indicating that
$\mathbf{S}$ is inconsistent. Alternatively, if one has rigorously
applied mass balancing for each chemical reaction whilst reconstructing
a network \citep{thieleTests,Thorleifsson2011}, one will be able
to assign a positive molecular mass corresponding to each of the molecules
in the reconstruction. This strictly positive molecular mass vector
satisfies \eqref{eq:consistencyTest2} and \eqref{eq:consistencyTest3},
which is sufficient to conclude that the corresponding stoichiometric
matrix is consistent.

The only other condition for Theorem \ref{thm:existenceSS} to hold
is that continuous kinetic rate laws be formulated in a manner such
that the concentration of any molecule can never be negative \citep{chellaboina2009modeling}.
In this paper, we have framed Theorem \ref{thm:existenceSS} in terms
of elementary reaction kinetics that satisfy concentration non-negativity
if the elementary kinetic parameters are non-negative. However, one
can envisage a more general version of Theorem \ref{thm:existenceSS}
as there are many other continuous kinetic rate laws that satisfy
concentration non-negativity \citep{liebermeister2010modular}. Any
continuous kinetic rate law, where the rate of a unidirectional reaction
is non-negative and zero if and only if any of the concentrations
of the molecules consumed in that reaction are zero \citep{chellaboina2009modeling},
would form the conditions for a generalised version of Theorem \ref{thm:existenceSS}.
In any case, any phenomenological kinetic rate law can be derived from
assumptions that allow simplification of a system of elementary chemical
reactions \citep{cook2007eka}. As such, phenomenological kinetic
modeling has its foundation in mass action kinetics.

\section{Conclusion}

It is 50 years since Wei's Axiom's on the existence of steady states
for chemical reaction systems\citep{wei1962axiomatic} and almost
40 years since Horn \& Jackson \citep{horn1972gma} considered what
they termed \emph{general mass action kinetics}. In Horn \& Jackson's
setting, elementary reaction rates are proportional to the abundance
of the substrates involved in the reaction, each to the power of the
absolute value of the corresponding stoichiometric coefficient. However,\emph{
general mass action kinetics} considers systems where kinetic parameters
need not be thermodynamically feasible, stoichiometric coefficients
need not be integers, and mass need not be conserved by each reaction.
With regard to modeling chemical reaction networks, we agree that
consideration of reactions with thermodynamically infeasible kinetic
parameters does seem profitable for representing the perpetual forcing
of a system, purely for modeling purposes. However, violation of mass
conservation appears unnecessary, as a pair of thermodynamically infeasible
kinetic parameters are sufficient to force a system away from equilibrium,
and counterproductive, as the resulting system may not admit a non-equilibrium
steady state. Horn \& Jackson \citep{horn1972gma} do realize that
the conditions for Wei's existence result are unmet when mass is not
conserved. We conclude that, rather than \emph{general mass action
kinetics,} assuming \emph{mass conserved elementary kinetics} is sufficient
for modeling non-equilibrium steady states in arbitrary large biochemical
networks, as one is then sure that at least one steady state does
actually exist. Similar conclusions hold for phenomenological kinetic
modeling, .e.g., with Michaelis-Menten kinetics, as long as the continuos
rate laws are such that concentration can never be negative.

\section*{Acknowledgements}

This work was supported by the U.S. Department of Energy (Offices
of Advanced Scientific Computing Research \& Biological and Environmental
Research) as part of the Scientific Discovery Through Advanced Computing
program (Grant No. DE-SC0002009). I.T. was also supported, in part,
by a Marie Curie International Reintegration Grant (No.~249261) within
the 7th European Community Framework Program.

\footnotesize
\frenchspacing

\bibliographystyle{jtb} 


\end{document}